\documentclass[sunil2,ChapterTOCs]{sunil} %See documentation for other class options
\usepackage{fixltx2e,fix-cm}
\usepackage{graphicx}
\usepackage{subfigure}
\usepackage{makeidx}
\usepackage{multicol}

\usepackage[utf8]{inputenc}
\usepackage[centertags]{amsmath}
\usepackage{amsthm,amsfonts}
\usepackage{amssymb}
\usepackage{epsfig}
\usepackage{latexsym}
\usepackage{caption}
\usepackage{indentfirst}
\usepackage{mathrsfs}
\usepackage{amsbsy}
\usepackage{bbding}
\usepackage{pdflscape}
\usepackage{enumitem}
\usepackage{setspace}
\usepackage{xcolor}

\newcommand{\ie}{\emph{i.e.}}

\newenvironment{pf}{\noindent\textbf{Proof.}\quad}{\hfill{$\Box$}}

\theoremstyle{plain}
\newtheorem{thm}{Theorem}[section]
\newtheorem{cor}[thm]{Corollary}
\newtheorem{lem}[thm]{Lemma}
\newtheorem{prop}[thm]{Proposition}

\theoremstyle{definition}
\newtheorem{defn}[thm]{Definition}
\newtheorem{rem}[thm]{Remark}
\newtheorem{ex}[thm]{Example}

\newcommand{\Fq}{\mathbb{F}_{q}}
\newcommand{\To}{\longrightarrow}
\newcommand{\Tr}{{\rm Tr}}

\newcommand{\Span}{{\rm Span}}

\newcommand{\C}{{\mathbb C}}

\newcommand{\E}{{\mathbb E}}
\newcommand{\F}{{\mathbb F}}
\newcommand{\G}{{\mathbb G}}
\newcommand{\HH}{{\mathbb H}}

\newcommand{\K}{\mathbb{K}}

\date{}
\begin{document}
\title{Quasi-Cyclic Codes}

\author{Cem G\"uneri, San Ling and Buket \"Ozkaya}
\maketitle

%\tableofcontents
\chapter{Quasi-Cyclic Codes}

\section{Introduction}\label{intro}
Cyclic codes are among the most useful and well-studied code families for various reasons, such as effective encoding and decoding. These desirable properties are due to the algebraic structure of cyclic codes, which is rather simple. Namely, a cyclic code can be viewed as an ideal in a certain quotient ring obtained from a polynomial ring with coefficients from a finite field.

It is then natural for coding theorists to search for generalizations of cyclic codes, which also have nice properties. This chapter is devoted to one of the first such generalizations, namely the family of quasi-cyclic (QC) codes. Algebraically, QC codes are modules rather than ideals. As desired, QC codes turned out to be a very useful generalization and their investigation continues intensively today.

There is a vast literature on QC codes and it is impossible to touch upon all aspects in a chapter. This chapter is constrained to some of the most fundamental aspects to which the authors also made contributions over the years. We start with the algebraic structure of QC codes, present how the vectorial and algebraic descriptions are related to each other, and also discuss the code generators. Then we focus on the decomposition of QC codes, via the Chinese Remainder Theorem (CRT decomposition) and concatenation. In fact, these decompositions turn out to be equivalent in some sense, which is also mentioned in the chapter. Decomposition of a QC code yields a trace representation. Moreover, the dual QC code can be described in the decomposed form, which has important consequences of its own. For instance, self-dual and linear complementary dual QC codes can be characterized in terms of their CRT decompositions. Moreover, the asymptotic results presented in the chapter heavily rely on the decomposed structure of QC codes. Three different general minimum distance bounds on QC codes are presented here, with fairly complete proofs or detailed elaborations. A relation between QC codes and convolutional codes is also shown in the end.

\section{Algebraic Structure}\label{background}
Let $\Fq$ denote the finite field with $q$ elements, where $q$ is a prime power, and let $m$ and $\ell$ be two positive integers. A linear code $C$ of length $m\ell$ over $\Fq$ is called a \emph{quasi-cyclic (QC) code of index $\ell$} if it is invariant under shift of codewords by $\ell$ positions and $\ell$ is the minimal number with this property. Note that, if $\ell=1$, then $C$ is a cyclic code. If we view codewords of $C$ as $m \times \ell$ arrays as follows 

\begin{equation}\label{array}
\mathbf{c}=\left(
  \begin{array}{ccc}
    c_{00} & \ldots & c_{0,\ell-1} \\
    \vdots &  & \vdots \\
    c_{m-1,0} & \ldots & c_{m-1,\ell-1} \\
  \end{array}
\right),
\end{equation} 
then being invariant under shift by $\ell$ units amounts to being closed under the row shift where each row is moved downward one row and the bottom row moved to the top.

Consider the principal ideal $I=\langle x^m-1 \rangle$ of $\Fq[x]$ and define the quotient ring $R:=\Fq[x]/I$. If $T$ represents the shift-by-1 operator on $\Fq^{m\ell}$, then its action on $\mathbf{v}\in \F_q^{m\ell}$ will be denoted by $T\cdot \mathbf{v}$. Hence, $\Fq^{m\ell}$ has an
$\Fq[x]$-module structure given by the multiplication
\begin{eqnarray*}
\Fq[x] \times \F_{q}^{m\ell} & \longrightarrow & \F_{q}^{m\ell}\\
(a(x),\mathbf{v}) & \mapsto & a(T^\ell)\cdot \mathbf{v} .
\end{eqnarray*}
Note that, for $a(x)=x^m-1$, we have $a(T^\ell)\cdot \mathbf{v} = (T^{m\ell})\cdot \mathbf{v} - \mathbf{v}=0.$ Hence, the ideal $I$ fixes $\Fq^{m\ell}$ and we can view $\F_q^{m\ell}$ as an $R$-module. Therefore, a QC code
$C\subseteq \F_q^{m\ell}$ of index $\ell$ is an $R$-submodule of $\F_q^{m\ell}$.

To an element $\mathbf{c}\in \Fq^{m\times \ell} \simeq \Fq^{m\ell}$ as in (\ref{array}), we associate an element of $R^\ell$ as
\begin{equation*} \label{associate-1}
\mathbf{c}(x):=(c_0(x),c_1(x),\ldots ,c_{\ell-1}(x)) \in R^\ell ,
\end{equation*}
where, for each $0\leq j \leq \ell-1$, 
\begin{equation*}\label{columns} 
c_j(x):= c_{0,j}+c_{1,j}x+c_{2,j}x^2+\cdots +c_{m-1,j}x^{m-1} \in R .
\end{equation*} 
Thus, the following map is an $R$-module isomorphism:
\begin{equation}\begin{array}{lll} \label{identification-1}
\phi: \hspace{2cm} \F_q^{m\ell} & \longrightarrow & R^\ell  \\
\mathbf{c}=\left(
  \begin{array}{ccc}
    c_{00} & \ldots & c_{0,\ell-1} \\
    \vdots &  & \vdots \\
    c_{m-1,0} & \ldots & c_{m-1,\ell-1} \\
  \end{array}
\right) & \longmapsto & \mathbf{c}(x) .
\end{array}
\end{equation}
Note that, for $\ell=1$, this amounts to the classical polynomial representation of cyclic codes. Observe that the $T^\ell$-shift on $\F_q^{m\ell}$ corresponds to the componentwise multiplication by $x$ in $R^\ell$. Therefore, a $q$-ary QC code $C$ of length $m\ell$ and index $\ell$ can be considered as an $R$-submodule in $R^{\ell}$.

Lally and Fitzpatrick proved in \cite{LF} that every quasi-cyclic code, viewed as an $R$-submodule in $R^{\ell}$, has a generating set in the form of a reduced Gr{\"o}bner basis. In order to explain their findings, we need to fix some further notation first.

Consider the following ring homomorphism:
\begin{eqnarray*}
\Psi \ :\ \Fq[x]^{\ell} &\longrightarrow& R^{\ell} \\\nonumber
(f_0(x),f_1(x),\ldots, f_{\ell-1}(x))  &\longmapsto& (f_0(x)+I,f_1(x)+I,\ldots ,f_{\ell-1}(x)+I).
\end{eqnarray*}
Given a QC code $C\subseteq R^{\ell}$, it follows that the preimage $\Psi^{-1}(C)=\widetilde{C}$ of $C$ in $\Fq[x]^{\ell}$ is an $\Fq[x]$-submodule containing $\widetilde{K} =\{(x^m-1)\mathbf{e}_j : 0\leq j \leq \ell-1\}$, where $\mathbf{e}_j$ denotes the standard basis vector of length $\ell$ with 1 at the coordinate $j$ and 0 elsewhere. Throughout, the tilde~$\widetilde{\phantom{c}}$\ represents structures over $\Fq[x]$.

Since $\widetilde{C}$ is a submodule of the finitely generated free module $\Fq[x]^{\ell}$ over the principal ideal domain $\Fq[x]$ and it contains $\widetilde{K}$, it has a generating set of the form $$\{\mathbf{u}_1,\ldots,\mathbf{u}_p, (x^m-1)\mathbf{e}_0,\ldots,(x^m-1)\mathbf{e}_{\ell-1}\},$$  where $p$ is a nonnegative integer and $\mathbf{u}_b = (u_{b,0}(x),\ldots,u_{b,\ell-1}(x))\in \Fq[x]^{\ell}$ for each $b \in \{1,\ldots,p\}$. Hence, the rows of the matrix
$$M=\left(\begin{array}{ccc}
    u_{1,0}(x) & \ldots & u_{1,\ell-1}(x) \\
    \vdots &  & \vdots \\
    u_{p,0}(x) & \ldots & u_{p,\ell-1}(x) \\
     x^m-1 & \ldots & 0 \\
    \vdots & \ddots & \vdots \\
    0 & \ldots & x^m-1 \\
  \end{array}
\right)$$
generate $\widetilde{C}$. By using elementary row operations, we may triangularize $M$ so that another generating set can be obtained from the rows of an upper-triangular $\ell \times \ell$ matrix with entries in $\Fq[x]$ as follows:
\begin{equation}\label{gen matrix}
\widetilde{G}(x)=\left(\begin{array}{cccc}
    g_{00}(x) & g_{01}(x) & \ldots & g_{0,\ell-1}(x) \\
    0 & g_{11}(x) & \ldots & g_{1,\ell-1}(x) \\
    \vdots & \vdots & \ddots & \vdots \\
    0 & 0 &\ldots & g_{\ell-1,\ell-1}(x)\\
  \end{array}
\right),
\end{equation}
where $\widetilde{G}(x)$ satisfies the following conditions (see \cite[Theorem 2.1]{LF}):
\begin{enumerate}[leftmargin=*]
    \item $g_{i,j}(x)=0$, for all $0\leq j < i \leq \ell-1$.
    \item deg$(g_{i,j}(x)) < $ deg$(g_{j,j}(x))$, for all $i <j$.
    \item $g_{j,j}(x) \mid (x^m-1)$, for all $0\leq j \leq \ell-1$.
    \item If $g_{j,j}(x) = x^m-1$, then $g_{i,j}(x) =0$, for all $i\neq j$.
\end{enumerate}
Note that the rows of $\widetilde{G}(x)$ are nonzero and each nonzero codeword of $\widetilde{C}$ can be expressed in the form $(0,\ldots,0,c_j(x),\ldots,c_{\ell-1}(x)),$ where $j\geq 0$, $ c_j(x)\neq 0$ and $g_{j,j}(x)\,|\,c_j(x)$. This implies that the rows of $\widetilde{G}(x)$ form a Gr\"obner basis of $\widetilde{C}$ with respect to the position-over-term (POT) order in $\Fq[x]$, where the standard basis vectors $\{\mathbf{e}_0,\ldots,\mathbf{e}_{\ell-1}\}$ and the monomials $x^i$ are ordered naturally in each component. Moreover, the second condition above implies that the rows of $\widetilde{G}(x)$ form a reduced Gr\"obner basis of $\widetilde{C}$, which is uniquely defined up to multiplication by constants with monic diagonal elements.

Let $G(x)$ now be the matrix with the rows of $\widetilde{G}(x)$ under the image of the homomorphism $\Psi$. Clearly, the rows of $G(x):=\widetilde{G}(x) \mod I$ form an $R$-generating set for $C$. When $C$ is the zero code of length $m\ell$, we have $p=0$ which implies $G(x)=\mathbf{0}_{\ell\times\ell}$. Otherwise we say that $C$ is an $r$-generator QC code (generated as an $R$-submodule) if $G(x)$ has $r$ (nonzero) rows. The $\Fq$-dimension of $C$ is given by (see \cite[Corollary 2.4]{LF} for the proof) $$m\ell-\sum_{j=0}^{\ell-1}\mbox{deg}(g_{j,j}(x))=\sum_{j=0}^{\ell-1}\left[m -\mbox{deg}(g_{j,j}(x))\right].$$

\section{Decomposition of Quasi-Cyclic Codes}\label{DFT}
From this section on, we assume that $\gcd(m,q)=1$. 

\subsection{The Chinese Remainder Theorem and Concatenated Decompositions of QC Codes} \label{Concatenation}
We now describe the decomposition of a QC code over $\F_q$ into shorter codes over extension fields of $\F_q$. We follow the brief presentation in \cite{GO} and refer the reader to \cite{LS} for details. Let the polynomial $x^m-1$ factor into irreducible polynomials in $\F_q[x]$ as
\begin{equation}\label{irreducibles}
x^m-1=f_1(x)f_2(x)\ldots f_s(x).
\end{equation}
Since $m$ is relatively prime to $q$, there are no repeating factors in (\ref{irreducibles}). By the Chinese Remainder Theorem (CRT), we have the following ring isomorphism:
\begin{equation} \label{CRT-1}
R\cong \bigoplus_{i=1}^{s} \F_q[x]/\langle f_i(x)\rangle .
\end{equation}
Since each $f_i(x)$ divides $x^m-1$, its roots are powers of some fixed primitive $m^{th}$ root of unity $\xi$ in an extension field of $\Fq$. For each $i=1,2,\ldots,s$, let $u_i$ be the smallest nonnegative integer such that $f_i(\xi^{u_i})=0$. Since the $f_i(x)$'s are irreducible, the direct summands in (\ref{CRT-1}) are field extensions of $\F_q$. If $\E_i:=\F_q[x]/\langle f_i(x) \rangle$ for $1\leq i \leq s$, then we have
\begin{eqnarray} \label{CRT-2}
R & \cong & \E_1 \oplus \cdots \oplus \E_{s}
 \nonumber \\
a(x) & \mapsto & \left(a(\xi^{u_1}),\ldots ,a(\xi^{u_s}) \right).
\end{eqnarray}
This implies that
\begin{equation*} \label{CRT-3}
R^{\ell}\cong \E_1^{\ell} \oplus \cdots  \oplus \E_{s}^{\ell}.
\end{equation*}
Hence, a QC code $C\subseteq R^\ell$ can be viewed as an $(\E_1\oplus \cdots \oplus \E_{s})$-submodule of $\E_1^{\ell} \oplus\cdots  \oplus \E_{s}^{\ell}$ and decomposes as
\begin{equation} \label{constituents}
C\cong C_1\oplus \cdots  \oplus C_{s},
\end{equation}
where $C_i$ is a linear code of length $\ell$ over $\E_i$, for each $i$. These length $\ell$ linear codes over various extensions of $\F_q$ are called the \emph{constituents} of $C$. Let $C\subseteq R^\ell$ be $r$-generated as an $R$-module by
$$\{\bigl(a_{1,0}(x),\ldots ,a_{1,\ell-1}(x)\bigr),\ldots ,
\bigl(a_{r,0}(x),\ldots ,a_{r,\ell-1}(x)\bigr)\} \subset R^\ell.$$
Then, for $1\leq i \leq s$, we have
\begin{equation}\label{explicit constituents}
 C_i  =  \Span_{\E_i}\bigl\{\bigl(a_{b,0}(\xi^{u_i}),\ldots
,a_{b,\ell-1}(\xi^{u_i})\bigr): 1\leq b \leq r \bigr\}.   
\end{equation}

%\subsection{Concatenated Structure of QC Codes} \label{Concatenation}
Note that each extension field $\E_i$ above is isomorphic to a minimal cyclic code of length $m$ over $\F_q$, namely, the cyclic code whose check polynomial is $f_i(x)$. If we denote by $\theta_i$ the generating primitive idempotent for the minimal cyclic code in consideration, then the isomorphism is given by the maps
\begin{eqnarray} \label{isoms}
\begin{array}{ccc} \varphi_i:\langle \theta_i \rangle
& \longrightarrow & \E_i \\ \hspace{0.5cm} a(x)& \longmapsto &
a(\xi^{u_i}) \end{array}
& \rm{and}& \begin{array}{ccc} \psi_i: \E_i & \longrightarrow & \langle \theta_i \rangle \\
\hspace{0.5cm} \delta & \longmapsto & \sum\limits_{k=0}^{m-1} a_kx^k
\end{array}\ \ ,
\end{eqnarray}
where
$$a_k=\frac{1}{m} \Tr_{\E_i/\F_q}(\delta\xi^{-ku_i}).$$
Here, $\Tr_{\E_i/\F_q}$ denotes the trace map from $\E_i$ onto $\F_q$. If $[\E_i:\F_q]=e_i$, then $\Tr_{\E_i/\F_q}(x)= x+x^q+\cdots+x^{q^{e_i-1}}$. If $C_i$ is a length $\ell$ linear code over $\E_i$, we denote its concatenation with $\langle \theta_i \rangle$ by $\langle \theta_i \rangle \Box C_i$ and the concatenation is carried out by the map $\psi_i$, extended to $\E_i^\ell$. In other words, $\psi_i$ is applied to each symbol of the codeword in $C_i$ to produce an element of $\langle \theta_i\rangle ^\ell$.

Jensen gave the following concatenated description for QC codes.
\begin{thm} [{\cite{J}}] \label{Jensen's thm}\hfill
\begin{itemize}[leftmargin=*]
\item[i.] Let $C$ be an $R$-submodule of $R^\ell$ (i.e., a QC code). Then for some subset $\mathcal{I}$ of $\{1,\ldots ,s\}$, there exist linear codes $C_i$ of length $\ell$ over $\E_i$, which can be explicitly described, such that $$C=\displaystyle\bigoplus_{i\in \mathcal{I}} \langle \theta_i \rangle \Box C_i.$$

\item[ii.] Conversely, let $C_i$ be a linear code in $\E_i^\ell$, for each $i\in \mathcal{I}\subseteq \{1,\ldots,s\}.$ Then, $$C=\displaystyle\bigoplus_{i\in \mathcal{I}} \langle \theta_i \rangle \Box C_i$$ is a $q$-ary QC code of length $m\ell$ and index $\ell$.
\end{itemize}
\end{thm}

It was proven in \cite[Theorem 4.1]{GO} that, for a given QC code $C$, the constituents $C_i$'s in (\ref{constituents}) and the outer codes $C_i$'s in the concatenated structure of Theorem~\ref{Jensen's thm} are equal to each other.

\subsection{Applications}\label{DFT application}

In this section, we present some constructions and characterizations of QC codes using their CRT decomposition.

\subsubsection{Trace Representation}\label{trace section}
By (\ref{isoms}) and Theorem \ref{Jensen's thm}, an arbitrary codeword $\mathbf{c}\in C$ can be written as an $m\times \ell$ array in the form (see \cite[Theorem 5.1]{LS})
\begin{equation}\label{QC tr}
\mathbf{c}=\frac{1}{m}\left(\begin{array}{c}\left(\sum\limits_{i=1}^{s}\Tr_{\E_{i}/\F_q}\left(\lambda_{i,t}\xi^{-0u_{i}} \right) \right)_{0\leq t \leq
\ell-1} \\
 \vdots \\
\left(\sum\limits_{i=1}^{s}\Tr_{\E_{i}/\F_q}\left(\lambda_{i,t}\xi^{-(m-1)u_{i}}
\right) \right)_{0\leq t \leq \ell-1}
\end{array} \right),
\end{equation} 
where $\lambda_i=(\lambda_{i,0},\ldots ,\lambda_{i,\ell-1}) \in C_i$, for all $i$. Since $mC=C$, every codeword in $C$ can still be written in the form of (\ref{QC tr}) with the constant $\frac{1}{m}$ removed. Note that the row shift invariance of codewords amounts to being closed under multiplication by $\xi^{-1}$ in this representation.

Let $\F:=\Fq(\xi^{u_1},\ldots,\xi^{u_s})$ be the splitting field of $x^m-1$ (\ie, the smallest field containing all the $\E_i$'s) and let $k_1,\ldots, k_s \in\F$ with $\Tr_{\F/\E_{i}}(k_i)=1$, for each $i$. We can now unify the traces and rewrite $\mathbf{c} \in C$ as follows (see \cite{GO2}):
\begin{equation}\label{QC trace} \mathbf{c}=\left(\left(\Tr_{\F/\F_q}\left(\sum\limits_{i=1}^{s}k_i\lambda_{i,t}\xi^{-j u_{i}}
\right) \right)_{\substack{0\leq t \leq \ell-1\\ 0 \leq j \leq m-1}
} \right).\end{equation} 
Note that the case $\ell=1$ gives us the trace representation of a cyclic code of length $m$, in the sense of the following formulation:

\begin{thm} [{\cite[Proposition 2.1]{W}}] \label{Wolfmann}
Let $\gcd(m,q)=1$ and let $\xi$ be a primitive $m^{th}$ root of unity in some extension $\F$ of $\F_q$. Assume that $u_1,\ldots ,u_s$ are nonnegative integers and let $C$ be a $q$-ary cyclic code of length $m$ such that the basic zero set of its dual is $BZ(C^\perp)=\left\{\xi^{u_1},\xi^{u_2},\ldots ,\xi^{u_s} \right\}$ (\ie, the generator polynomial of $C^\perp$ is $\prod_{i} m_{i}(x),$ where $m_{i}(x)\in \F_q[x]$ is the minimal polynomial of $\xi^{u_i}$ over $\F_q$). Then
\begin{equation*} \label{trace for cyclic}
C=\bigl\{\left(\Tr_{\F/\F_q}(c_1\xi^{j u_1}+\cdots + c_s\xi^{j u_s}
)\right)_{0\leq j \leq m-1}: c_1,\ldots ,c_s \in \F  \bigr\}.
\end{equation*}
\end{thm}
Hence, the columns of any codeword in the QC code $C$ viewed as in (\ref{QC trace}) lie in the cyclic code $D\subseteq\Fq^m$ with $\mbox{BZ}(D^\perp)=\left\{\xi^{-u_1},\ldots ,\xi^{-u_s} \right\}$ (see \cite[Proposition 4.2]{GO2}).

\begin{ex}\label{cubic cons}
Let $m=3$ and $q \equiv 2 \pmod 3$ such that $x^3-1$ factors into $x-1$ and $x^2+x+1$ over $\Fq$. By (\ref{CRT-1}) and (\ref{CRT-2}), we obtain $$R=\F_q[x]/\langle x^3-1\rangle\cong  \F_q[x]/\langle x-1\rangle \oplus \F_q[x]/\langle x^2+x+1\rangle \cong \F_q \oplus \F_{q^2}.$$
Therefore, any QC code $C$ of length $3\ell$ and index $\ell$ has two linear constituents $C_1 \subseteq \Fq^{\ell}$ and $C_2 \subseteq \F_{q^2}^{\ell}$ such that $C\cong C_1 \oplus C_2$. By (\ref{QC trace}), the codewords of $C$ are of the form (see \cite[Theorem 6.7]{LS}) 
$$\left(\begin{array}{c} \mathbf{z}+2\mathbf{a}-\mathbf{b}\\
\mathbf{z}-\mathbf{a}+2\mathbf{b}\\
\mathbf{z}-\mathbf{a}-\mathbf{b}\end{array} \right),$$ where $\mathbf{z} \in C_1$, $\mathbf{a}+\beta \mathbf{b} \in C_2$ ($\mathbf{a},\mathbf{b} \in \Fq^{\ell}$) and $\beta\in\F_{q^2}$ such that $\beta^2+\beta+1=0$.

In particular, let $q=2^t$, where $t$ is odd (otherwise, $2^t\equiv 1\hspace{-4pt} \pmod 3$ if $t$ is even and $x^3-1$ splits into linear factors). By using the fact $2\mathbf{a}=2\mathbf{b}=\mathbf{0}_{\ell}$ in even characteristic, we can simplify the above expression further and write $C$ as 
\begin{equation} \label{cubic}
C=\left\{
\left(\begin{array}{c}\mathbf{z}+\mathbf{b}\\
\mathbf{z}+\mathbf{a}\\
\mathbf{z}+\mathbf{a}+\mathbf{b}\end{array} \right)\ :\ \mathbf{z} \in C_1, \mathbf{a}+\beta \mathbf{b} \in C_2\right\}.
\end{equation}
\end{ex}

\begin{ex}
Let $m=5$ and let $q$ be such that $x^4+x^3+x^2+x+1$ is irreducible over $\Fq$. Let $\alpha\in\F_{q^4}$ such that $\alpha^4+\alpha^3+\alpha^2+\alpha+1=0$, and let $\Tr=\Tr_{\F_{q^4}/\Fq}$. Let $C_1 \subseteq \Fq^{\ell}$ and $C_2 \subseteq \F_{q^4}^{\ell}$ be two linear codes. Then the code (see \cite[Theorem 6.14]{LS}) 
$$C=\left\{\left(\begin{array}{c} 
\mathbf{z}+\Tr(\mathbf{y})\\
\mathbf{z}+\Tr(\alpha^{-1}\mathbf{y})\\
\mathbf{z}+\Tr(\alpha^{-2}\mathbf{y})\\
\mathbf{z}+\Tr(\alpha^{-3}\mathbf{y})\\
\mathbf{z}+\Tr(\alpha^{-4}\mathbf{y})\end{array} \ :\ \mathbf{z} \in C_1,\mathbf{y} \in C_2\right)\right\}$$ is a QC code of length $5\ell$ and index $\ell$ over $\Fq$.

In particular, let $q=2^t$ and set $\mathbf{y}= \mathbf{c}+\alpha \mathbf{d}+\alpha^2 \mathbf{e}+\alpha^3 \mathbf{f}$, for some $\mathbf{c},\mathbf{d},\mathbf{e},\mathbf{f} \in \Fq^{\ell}$. We can rewrite the codewords in $C$ as: $$\left(\begin{array}{c} \mathbf{z}+\mathbf{d}+\mathbf{e}+\mathbf{f}\\
\mathbf{z}+\mathbf{c}+\mathbf{e}+\mathbf{f}\\
\mathbf{z}+\mathbf{c}+\mathbf{d}+\mathbf{f}\\
\mathbf{z}+\mathbf{c}+\mathbf{d}+\mathbf{e}\\
\mathbf{z}+\mathbf{c}+\mathbf{d}+\mathbf{e}+\mathbf{f}\end{array} \right),$$ where $\mathbf{z} \in C_1, \mathbf{c}+\alpha \mathbf{d}+\alpha^2 \mathbf{e}+\alpha^3 \mathbf{f} \in C_2$.
\end{ex}

\subsubsection{Self-Dual and Complementary-Dual QC Codes} \label{SD-LCD}
Observe that the monic polynomial $x^{m}-1$ is self-reciprocal. We rewrite the factorization into irreducible polynomials given in (\ref{irreducibles}) as follows, which is needed for the dual code analysis (see also \cite{LS}):
\begin{equation}\label{irreducibles-2}
x^{m}-1=g_{1}(x)\cdots g_{n}(x) h_{1}(x)h_{1}^*(x)\cdots
h_{p}(x)h_{p}^*(x),
\end{equation}
where $g_{i}(x)$ is self-reciprocal, for all $1\leq i\leq n$, and $h_{t}^*(x)=x^{\deg(h_t)}h_{t}(x^{-1})$ denotes the reciprocal of $h_{t}(x)$, for all $1\leq i\leq n$ and $1\leq t\leq p$. Note that $s=n+2p$, since (\ref{irreducibles}) and (\ref{irreducibles-2}) must be identical.

Let $\G_{i}:=\F_q[x]/\langle g_{i}(x) \rangle$, $\HH_{t}':=\F_q[x]/\langle h_{t}(x) \rangle$ and $\HH_{t}'':=\F_q[x]/\langle h_{t}^*(x) \rangle$, for each $i$ and $t$. By the CRT, the decomposition of $R$ in (\ref{CRT-1}) now becomes:
\begin{equation*} \label{CRT-4}
R \cong \left( \bigoplus_{i=1}^{n} \G_{i} \right) \oplus \left( \bigoplus_{t=1}^{p} \Bigl( \HH_{t}' \oplus \HH_{t}'' \Bigr) \right),
\end{equation*}
which implies
\begin{equation*}
R^{\ell}\cong \left(\bigoplus_{i=1}^{n} \G_i^{\ell}\right) \oplus \left(\bigoplus_{t=1}^{p} (\HH'_{t})^{\ell} \oplus (\HH''_{t})^{\ell}\right).
\end{equation*}

Hence, a QC code $C\subseteq R^{\ell}$ viewed as an $R$-submodule of $R^{\ell}$ now decomposes as (cf. (\ref{constituents}))
\begin{equation} \label{constituents-2}
C\cong\left( \bigoplus_{i=1}^{n} C_i \right) \oplus \left( \bigoplus_{t=1}^{p} \Bigl( C_{t}' \oplus C_{t}'' \Bigr) \right),
\end{equation}
where $C_i$'s are the $\G_i$-linear constituents of $C$ of length $\ell$, for all $i=1,\ldots,n$, $C'_{t}$'s and $C''_{t}$'s are the $\HH'_t$-linear and $\HH''_t$-linear constituents of $C$ of length $\ell$, respectively, for all $t=1,\ldots,p$. By fixing roots corresponding to irreducible factors $g_i(x),h_t(x),h_t^*(x)$ and via the CRT, one can write explicitly the constituents, as in (\ref{explicit constituents}), in this setting as well. Observe that $\HH'_t=\HH''_t$, since $\Fq(\xi^a)=\Fq(\xi^{-a})$, for any $a\in \{0,1,\ldots,m-1\}$.

Since $g_i(x)$ is self-reciprocal, the cardinality of $\G_i$, say $q_i$, is an even power of $q$ for all $1\leq i \leq n$ with two exceptions. One of these exceptions, for all $m$ and $q$, is the field coming from the irreducible factor $x-1$ of $x^m-1$. When $q$ is odd and $m$ is even, $x+1$ is another self-reciprocal irreducible factor of $x^m-1$. In these cases, $q_i=q$. Except for these two cases, we equip each $\G_i^{\ell}$ with the inner product
\begin{equation}
\label{hermprod}
\langle\, \mathbf{c},\mathbf{d}\,\rangle := \sum_{j=0}^{\ell-1} c_{j} d_{j}^{\sqrt{q_i}},
\end{equation}
where $\mathbf{c}=(c_{0},\ldots ,c_{\ell-1}), \mathbf{d}=(d_{0},\ldots ,d_{\ell-1})\in \G_i^{\ell}$. This is the Hermitian inner product. For the two exceptions, in which case the corresponding field $\G_i$ is $\mathbb{F}_q$, we equip $\G_i^{\ell}$ with the usual Euclidean inner product. With the appropriate inner product, Hermitian or Euclidean, on $\G_i$, $\bot_{\G}$ denotes the dual on $\G_i^{\ell}$. For each $1\leq t \leq p$, $\HH_t'^{\ell}=\mathbb{H}_t''^{\ell}$ is also equipped with the Euclidean inner product;  $\bot_e$ denotes the dual on $\HH_t'^{\ell}=\mathbb{H}_t''^{\ell}$. 

The dual of a QC code is also QC and the following result is immediate.
\begin{prop}[{\cite[Theorem 4.2]{LS}}]\label{duality}
Let $C$ be a QC code with CRT decomposition as in (\ref{constituents-2}). Then its dual code $C^{\bot}$ (under the Euclidean inner product) is of the form
\begin{equation}\label{dual}
C^\bot=\left( \bigoplus_{i=1}^n C_i^{\bot_{\G}} \right) \oplus \left( \bigoplus_{t=1}^p \Bigl( C_t''^{\bot_e} \oplus C_t'^{\bot_e} \Bigr) \right).
\end{equation}
\end{prop}

Recall that a linear code $C$ is said to be self-dual if $C=C^\bot$ and $C$ is called a linear code with complementary dual (LCD) if $C \cap C^\bot = \{\mathbf{0}\}$. We now characterize self-dual and LCD QC codes (QCCD) via their constituents (see \cite{GOS,LS2}).

\begin{thm} \label{SD-CDcriteria}
Let $C$ be a $q$-ary QC code of length $m\ell$ and index $\ell$, which has a CRT decomposition as in (\ref{constituents-2}).
\begin{enumerate}[leftmargin=*]
\item[i.] $C$ is (Euclidean) self-dual if and only if $C_i=C_i^{ \bot_{\G}}$, for all $1\leq i \leq n$, and $C_t'' = C_t'^{\bot_e}$ over $\HH_t'=\HH_t''$, for all $1\leq t \leq p$.
\item[ii.] $C$ is (Euclidean) LCD if and only if $C_i\cap C_i^{\bot_{\G}}=\{\mathbf{0}\}$, for all $1\leq i \leq n$, and $C_t' \cap C_t''^{\bot_e}=\{\mathbf{0}\}$, $C_t'' \cap C_t'^{\bot_e}=\{\mathbf{0}\}$ over $\HH_t'=\HH_t''$, for all $1\leq t \leq p$.
\end{enumerate}
\end{thm}

\begin{proof}
The proof is immediate from the CRT decompositions of $C$ in (\ref{constituents-2}) and of its dual $C^{\bot}$ in (\ref{dual}).
\end{proof}
The following special cases are easy to derive from Theorem \ref{SD-CDcriteria} above.

\begin{cor} \label{CDinstance} \hfill
\begin{enumerate}[leftmargin=*]
\item[i.] If the CRT decomposition of $C$ is as in (\ref{constituents-2}) with $C_i=C_i^{ \bot_{\G}}$, for all $1\leq i \leq n$, and $C_t'=\{0\}, C_t''=(\HH_t'')^{\ell}$ (or $C_t''=\{0\}, C_t'=(\HH_t')^{\ell}$), for all $1\leq t \leq p$, then $C$ is self-dual.
\item[ii.] If the CRT decomposition of $C$ is as in (\ref{constituents-2}) with $C_i\cap C_i^{\bot_{\G}}=\{\mathbf{0}\}$, for all $1\leq i \leq n$, and Euclidean LCD codes $C_t'=C_t''$ over $\HH_t'=\HH_t''$, for all $1\leq t \leq p$, then $C$ is LCD.
\end{enumerate}
\end{cor}

\section{Minimum Distance Bounds}\label{distance section}

Given a $q$-ary code $C$ that is not the trivial zero code, the minimum (Hamming) distance of $C$ is defined as $d(C):=\min\{\mbox{wt}(\mathbf{c}) : 0\neq \mathbf{c}\in C\}$, where $\mbox{wt}(\mathbf{c})$ denotes the number of nonzero entries in $\mathbf{c}\in C$. Using the tools described in Sections \ref{background} and \ref{DFT}, we address some lower bounds on the minimum distance of a given QC code in this section.

\subsection{The Jensen Bound}
The concatenated structure in Theorem \ref{Jensen's thm} yields a minimum distance bound for QC codes (\cite[Theorem 4]{J}). This bound is a consequence of the work of Blokh and Zyablov in \cite{BZ}, which holds for general concatenated codes.

\begin{thm}\label{jensen}
Let $C$ be a $q$-ary QC code of length $m\ell$, index $\ell$ with the concatenated structure
$$C=\bigoplus\limits_{t=1}^{g} \langle \theta_{t} \rangle \Box C_{t},$$
where, without loss of generality, we have chosen the indices $t$ such that $0 < d(C_{1}) \leq d(C_{2}) \leq \cdots \leq d(C_{g})$. Then we have
\begin{equation*} \label{jensens concat bound}
d(C)\geq \min\limits_{1\leq e \leq g} \left\{d(C_{e}) d\Bigl(\langle \theta_{1} \rangle  \oplus \cdots \oplus \langle \theta_{e} \rangle \Bigr) \right\}.
\end{equation*}
\end{thm}

\begin{pf}
Let $\mathbf{c}$ be a nonzero codeword in $C$. By the concatenated structure of Section \ref{Concatenation} and (\ref{isoms}), 
$$\mathbf{c}=\sum_{i=1}^g \psi_i(\mathbf{c}_i),$$
with $\mathbf{c}_i\in C_i$, for all $1\leq i \leq g$. Assume that $e$ is the maximal index such that $\mathbf{c}_e$ is a nonzero vector and $\mathbf{c}_{e+1}, \ldots ,\mathbf{c}_g$ are all zero vectors. Hence,
$$\mathbf{c}=\psi_1(\mathbf{c}_1)+\cdots + \psi_e(\mathbf{c}_e),$$
and each $\psi_i(\mathbf{c}_i) \in \langle \theta_i\rangle ^\ell$. Since the sum of $\langle \theta_i \rangle$'s is direct, the codeword $\mathbf{c}$ belongs to $\left(\langle \theta_1 \rangle \oplus \cdots \oplus \langle \theta_e \rangle \right)^\ell$. Note that $\mathbf{c}_e$ has at least $d(C_e)$ nonzero coordinates, all of which are mapped to a nonzero vector in $\langle \theta_e \rangle$ by $\psi_e$. This implies that $\mathbf{c}$ has at least $d(C_e)$ nonzero coordinates, since each coordinate of $\mathbf{c}$ belongs to the direct sum $\langle \theta_1 \rangle \oplus \cdots \oplus \langle \theta_e \rangle$. Hence the result follows. 
\end{pf}

\subsection{The Lally Bound}

Let $C$ be a QC code of length $m\ell$ and index $\ell$ over $\Fq$. Let $\{1, \alpha, \ldots, \alpha^{\ell-1}\}$ be some fixed choice of basis of $\F_{q^\ell}$ as a vector space over $\Fq$. We view the codewords of $C$ as $m\times \ell$ arrays as in (\ref{array}) and consider the following map:
\begin{equation*}\begin{array}{lll} 
\Phi : \hspace{2cm} \F_q^{m\ell} & \longrightarrow & \hspace{.6cm}\F_{q^\ell}^m  \\[8pt]
\mathbf{c}=\left(
  \begin{array}{ccc}
    c_{0,0} & \ldots & c_{0,\ell-1} \\
    \vdots &  & \vdots \\
    c_{m-1,0} & \ldots & c_{m-1,\ell-1} \\
  \end{array}
\right) & \longmapsto &   \left(\begin{array}{c}
    c_{0} \\
    \vdots \\
    c_{m-1} \\
  \end{array}\right) ,
\end{array}
\end{equation*}
where $c_i=c_{i,0} + c_{i,1}\alpha + \cdots + c_{i,\ell-1}\alpha^{\ell-1}\in \F_{q^\ell}$, for all $0\leq i \leq m-1$.

Clearly, $\Phi(\mathbf{c})$ lies in some cyclic code, for any $\mathbf{c}\in C$. We now define the smallest such cyclic code as $\widehat{C}$. First, we equivalently extend the map $\Phi$ above to the polynomial description of codewords as in (\ref{identification-1}):
\begin{eqnarray}\label{Lal}
\Phi : \Fq[x]^{\ell} &\To& \F_{q^{\ell}}[x] \\
\mathbf{c}(x)=(c_0(x),\ldots,c_{\ell-1}(x)) &\longmapsto& c(x)=\displaystyle\sum_{j=0}^{\ell-1}c_j(x)\alpha^j .\notag
\end{eqnarray}
If $C$ has generating set $\{\mathbf{f}_1,\ldots,\mathbf{f}_r\}$, where $\mathbf{f}_k = (f^{(k)}_{0}(x),\ldots,f^{(k)}_{\ell-1}(x))\in \Fq[x]^{\ell}$, for each $k \in \{1,\ldots,r\}$, then  $\widehat{C}=\langle \gcd(f_1(x),\ldots, f_r(x), x^m-1)\rangle$ such that $f_k=\Phi(\mathbf{f}_k)\in \F_{q^{\ell}}[x]$, for all $k\in \{1,\ldots,r\}$ \cite{L2}.

Next, we consider the $q$-ary linear code of length $\ell$ that is generated by the rows of the codewords in $C$, which are represented as $m\times \ell$ arrays as in (\ref{array}). Namely, let $B$ be the linear block code of length $\ell$ over $\Fq$, generated by $\{\mathbf{f}_{k,i} : k\in\{1,\ldots,r\},\ i\in\{0,\ldots,m-1\}\}\subseteq \F_{q}^{\ell}$, where each $\mathbf{f}_{k,i}:=(f^{(k)}_{i,0},\ldots,f^{(k)}_{i,\ell-1}) \in\Fq^\ell$ is the vector of the $i^{th}$ coefficients of the polynomials $f^{(k)}_{j} (x)=f^{(k)}_{0,j}+f^{(k)}_{1,j}x+\cdots+f^{(k)}_{m-1,j}x^{m-1}$, for all $k\in\{1,\ldots,r\}$ and $j\in\{0,\ldots,\ell-1\}$.

Since the image of any codeword $\mathbf{c}(x) \in C$ under the map $\Phi$ is an element of the cyclic code $\widehat{C}$ over $\F_{q^\ell}$, there are at least $d(\widehat{C})$ nonzero rows in each nonzero codeword of $C$. For any $i\in\{0,\ldots,m-1\}$, the $i^{th}$ row $\mathbf{c}_i=(c_{i,0},\ldots,c_{i,\ell-1})$ of a codeword $\mathbf{c}\in C$ can be viewed as a codeword in $B$, therefore, a nonzero $\mathbf{c}_i$ has weight at least $d(B)$. Hence, we have shown the following.

\begin{thm}[Lally bound {\cite[Theorem 5]{L2}}]
Let $C$ be an $r$-generator QC code of length $m\ell$ and index $\ell$ over $\F_{q}$ with generating set $\{\mathbf{f}_1,\ldots,\mathbf{f}_r\}\subseteq \Fq[x]^{\ell}$. Let the cyclic code $\widehat{C}\subseteq \F_{q^\ell}^m$ and the linear code $B\subseteq\Fq^{\ell}$ be defined as above. Then, \[d(C)\geq d(\widehat{C})d(B).\] 
\end{thm}

\subsection{Spectral Bounds}

In \cite{ST}, Semenov and Trifonov developed a spectral theory for QC codes by using the upper triangular polynomial matrices (\ref{gen matrix}) given by Lally and Fitzpatrick \cite{LF}; this gives rise to a BCH-like minimum distance bound. Their bound was improved by Zeh and Ling in \cite{LZ} by using the Hartmann--Tzeng (HT) Bound for cyclic codes.\footnote{The BCH Bound is found in \cite{BC,H}; the HT Bound, which generalizes the BCH Bound, is found in \cite[Theorem 2]{HT}.} Before proceeding to the spectral theory of QC codes and their minimum distance bounds, we explore relevant minimum distance bounds for cyclic codes that generalize the BCH and HT Bounds, namely the Roos Bound and the shift bound.

\subsubsection{Cyclic Codes and Distance Bounds From Their Zeros}\label{prepa}
Recall that a cyclic code $C\subseteq \Fq^m$ can be viewed as an ideal of $R$. Since $R$ is a principal ideal ring, there exists a unique monic polynomial $g(x)\in R$ such that $C=\langle g(x)\rangle$, \ie, each codeword $c(x)\in C$ is of the form $c(x)=a(x)g(x)$, for some $a(x)\in R$. The polynomial $g(x)$, which is a divisor of $x^m-1$, is called the {\it generator polynomial} of $C$. For any positive integer $p$, let $\mathbf{0}_p$ denote throughout the all-zero vector of length $p$. A cyclic code $C=\langle g(x) \rangle$ is $\{\mathbf{0}_m\}$ if and only if $g(x)=x^m-1$. 

The roots of $x^m-1$ are of the form $1, \xi, \ldots, \xi^{m-1}$, where $\xi$ is a fixed primitive $m^{th}$ root of unity. Henceforth, let $\Omega :=\{\xi^k\ :\ 0\leq k \leq m-1\}$ be the set of all $m^{th}$ roots of unity. Recall that the splitting field of $x^m-1$ is denoted by $\F$, hence, $\F$ is the smallest extension of $\Fq$ that contains $\Omega$. Given the cyclic code $C=\langle g(x)\rangle$, the set $L:=\{\xi^k\ :\ g(\xi^k)=0\}\subseteq \Omega$ of roots of its generator polynomial is called the {\it zeros} of $C$. Note that $\xi^k\in L$ implies $\xi^{qk}\in L$, for each $k$. A nonempty subset $E\subseteq\Omega$ is said to be {\it consecutive} if there exist integers $e,n$ and $\delta$ with $e\geq 0,\delta \geq 2, n> 0$ and $\gcd(m,n)=1$ such that
\begin{equation} \label{cons zero set}
E:=\{\xi^{e+zn}\ :\ 0\leq z\leq \delta-2\}\subseteq\Omega.
\end{equation}

We now describe the Roos bound for cyclic codes. For $\emptyset\neq P\subseteq\Omega$, let $C_P$ denote any cyclic code of length $m$ over $\F_q$, whose set of zeros contains $P$. Let $d_P$ denote the least possible amount that the minimum distance of any $C_P$ can have. In particular, $d_P=d(C_P)$ when $P$ is the set of zeros of $C_P$.

\begin{thm}[Roos bound {\cite[Theorem 2]{R2}}] \label{Roos} Let $N$ and $M$ be two nonempty subsets of $\Omega$. If there exists a consecutive set $M'$ containing $M$ such that $|M'| \leq |M| + d_N -2$, then we have $d_{MN}\geq |M| + d_N -1$, where $MN:=\{\varepsilon\vartheta\mid \varepsilon\in M,\ \vartheta\in N\}$.
\end{thm}
If $N$ is consecutive like in (\ref{cons zero set}), then we get the following.
\begin{cor} [{\cite[Corollary 1]{R2}}] \label{Roos2}
Let $N, M$ and $M'$ be as in Theorem \ref{Roos}, with $N$ consecutive. Then $|M'| < |M| + |N|$ implies $d_{MN}\geq |M| + |N|$.
\end{cor}

\begin{rem}\label{Roos remark}
In particular, the case $M=\{1\}$ in Corollary~\ref{Roos2} yields the BCH bound for the associated cyclic code. Taking $M'=M$ in the corollary yields the HT bound. 
\end{rem}

Another improvement to the HT bound for cyclic codes, known as the shift bound, was given by van Lint and Wilson in \cite{LW}. To describe the shift bound, we need the notion of an {\it independent set}, which can be constructed over any field in a recursive way.

Let $S$ be a subset of some field $\K$ of any characteristic. One inductively defines a family of finite subsets of $\K$, called independent with respect to $S$, as follows:
\begin{enumerate}[leftmargin=*]
\item $\emptyset$ is independent with respect to $S$.
\item If $A \subseteq S$ is independent with respect to $S$, then $A\cup\{b\}$ is independent with respect to $S$ for all $b\in \K \setminus S$. \label{cond:2}
\item If $A$ is independent with respect to $S$ and $c\in\K^{*}$, then $cA$ is independent with respect to $S$. \label{cond:3}
\end{enumerate}

We define the {\it weight} of a polynomial $f(x) \in \K[x]$, denoted by $\mbox{wt}(f)$, as the number of nonzero coefficients in $f(x)$.

\begin{thm} [Shift bound {\cite[Theorem 11]{LW}}]\label{shift bound}
Let $0\neq f(x)\in \K[x]$ and let $S:=\{\theta\in \K\ : \ f(\theta)=0\}$. Then $\emph{wt}(f)\geq |A|,$ for every subset $A$ of $\K$ that is independent with respect to $S$.
\end{thm}

The minimum distance bound for a given cyclic code follows by considering the weights of its codewords $c(x)\in C$ and the independent sets with respect to subsets of its zeros $L$. Observe that, in this case, the universe of the independent sets is $\Omega$, not $\F$, because all of the possible roots of the codewords are contained in $\Omega$. Moreover, we choose $b$ from $\Omega \setminus P$ in Condition \ref{cond:2} above, where $P\subseteq L$, and $c$ in Condition \ref{cond:3} is of the form $\xi^k\in\F^{*}$, for some $0\leq k\leq m-1$. 

\begin{rem}\label{shift remark}
In particular, $A=\{\xi^{e+zn}\ :\ 0\leq z\leq \delta-1\}$ is independent with respect to the consecutive set $E$ in (\ref{cons zero set}), which gives the BCH bound for $C_E$. Let $D := \{\xi^{e + z  n_1 + y  n_2} : 0\leq z \leq \delta-2, 0\leq y\leq s\}$, for integers $b \geq 0$, $\delta \geq 2$ and positive integers $s, n_1$ and $n_2$ such that $\gcd(m, n_1) = 1$ and $\gcd(m, n_2) < \delta$. Then, for any fixed $\zeta\in \{0,\ldots,\delta-2\}$, $A_{\zeta}:=\{\xi^{e + z  n_1} :	0\leq z \leq\delta-2\}\cup\{\xi^{e + \zeta  n_1 + y  n_2} : 0\leq y \leq s+1\}$ is independent with respect to $D$ and we get the HT bound for $C_D$.
\end{rem}

\subsubsection{Spectral Theory of QC Codes}
Semenov and Trifonov \cite{ST} used the polynomial matrix $\widetilde{G}(x)$ in (\ref{gen matrix}) to develop a spectral theory for QC codes, a topic we now explore.

Given a QC code $C\subseteq R^{\ell}$, let the associated $\ell \times \ell$ upper triangular matrix $\widetilde{G}(x)$ be as in (\ref{gen matrix}) with entries in $\Fq[x]$. The {\it determinant} of $\widetilde{G}(x)$ is $$\det(\widetilde{G}(x)):=\prod_{j=0}^{\ell-1}g_{j,j}(x)$$
and we define an {\it eigenvalue} $\beta$ of $C$ to be a root of $\det(\widetilde{G}(x))$. Note that, since $g_{j,j}(x)\mid (x^m-1)$, for each $0\leq j\leq \ell-1$, all eigenvalues are elements of $\Omega$, \ie, $\beta=\xi^k$ for some $k\in\{0,\ldots,m-1\}$. The {\it algebraic multiplicity} of $\beta$ is the largest integer $a$ such that $(x-\beta)^a\mid \det(\widetilde{G}(x))$. The {\it geometric multiplicity} of $\beta$ is the dimension of the null space of $\widetilde{G}(\beta)$. This null space, denoted by $\mathcal{V}_{\beta}$, is called the {\it eigenspace} of $\beta$. In other words, 
\[
\mathcal{V}_{\beta}:=\big\{\mathbf{v}\in\F^{\ell} : \widetilde{G}(\beta)\mathbf{v}^{\top}=\mathbf{0}^{\top}_{\ell}\big\},
\]
where $\F$ is the splitting field of $x^m-1$, as before, and $^\top$ denotes the transpose. It was shown in \cite{ST} that, for a given QC code and the associated $\widetilde{G}(x) \in \Fq[x]^{\ell\times\ell}$, the algebraic multiplicity $a$ of an eigenvalue $\beta$ is equal to its geometric multiplicity $\dim_{\F}(\mathcal{V}_{\beta})$. 
\begin{lem}[{\cite[Lemma 1]{ST}}]\label{multiplicity lemma}
The algebraic multiplicity of any eigenvalue of a QC code $C$ is equal to its geometric multiplicity.
\end{lem}

From this point on, we let $\overline{\Omega}\subseteq \Omega$ denote the nonempty set of all eigenvalues of $C$ such that $|\overline{\Omega}|=t>0$. Note that $\overline{\Omega}=\emptyset$ if and only if the diagonal elements $g_{j,j}(x)$ in $\widetilde{G}(x)$ are constant and $C$ is the trivial full space code. Choose an arbitrary eigenvalue $\beta_i\in\overline{\Omega}$ with multiplicity $n_i$ for some $i \in\{1,\ldots,t\}$. Let $\{\mathbf{v}_{i,0},\ldots,\mathbf{v}_{i,n_i-1}\}$ be a basis for the corresponding eigenspace $\mathcal{V}_i$. Consider the matrix  
\begin{equation}\label{Eigenspace} 
 V_i:=\begin{pmatrix}
\mathbf{v}_{i,0} \\
\vdots\\
\mathbf{v}_{i,n_i-1} 
\end{pmatrix}
=
\begin{pmatrix}
v_{i,0,0}&\ldots&v_{i,0,\ell-1} \\
\vdots & \vdots & \vdots\\
v_{i,n_i-1,0}&\ldots&v_{i,n_i-1,\ell-1}
 \end{pmatrix},
\end{equation} 
having the basis elements as its rows. We let
\[
H_i:=(1, \beta_i,\ldots,\beta_i^{m-1})\otimes V_i \mbox{ and }
\]
\begin{equation}\label{parity check matrix} 
H:=\begin{pmatrix}
H_1 \\
\vdots\\
H_t 
\end{pmatrix}=\begin{pmatrix}
V_1&\beta_1 V_1 & \ldots &(\beta_1)^{m-1} V_1 \\
\vdots &\vdots & \vdots & \vdots\\
V_t & \beta_t V_t & \ldots & (\beta_t)^{m-1} V_t
\end{pmatrix}.
\end{equation} 
Observe that $H$ has $n:=\sum_{i=1}^t n_i$ rows. We also have that $n=\sum_{j=0}^{\ell-1}\mbox{deg}(g_{j,j}(x))$ by Lemma \ref{multiplicity lemma}. The rows of $H$ are linearly independent, a fact shown in \cite[Lemma 2]{ST}. This proves the following lemma.
\begin{lem}\label{rank lemma}
The matrix $H$ in (\ref{parity check matrix}) has rank $m\ell -\dim_{\Fq}(C)$.
\end{lem}

It is immediate to confirm that $H \mathbf{c}^{\top}=\mathbf{0}^{\top}_n$ for any codeword $\mathbf{c}\in C$. Together with Lemma \ref{rank lemma}, we obtain the following.

\begin{prop}[{\cite[Theorem 1]{ST}}]
The $n\times m\ell$ matrix $H$ in (\ref{parity check matrix}) is a parity-check matrix for $C$.
\end{prop}

\begin{rem}
Note that if $\overline{\Omega}=\emptyset$, then the construction of $H$ in (\ref{parity check matrix}) is impossible. Hence, we have assumed $\overline{\Omega}\neq\emptyset$ and we can always say $H=\mathbf{0}_{m\ell}$ if $C=\Fq^{m\ell}$. The other extreme case is when $\overline{\Omega}=\Omega$. By using Lemma \ref{rank lemma} above, one can easily deduce that a given QC code $C=\{\mathbf{0}_{m\ell}\}$ if and only if $\overline{\Omega}=\Omega$, each $\mathcal{V}_i=\F^{\ell}$ (equivalently, each $V_i=I_{\ell}$, where $I_{\ell}$ denotes the $\ell\times\ell$ identity matrix) and $n=m\ell$ so that we obtain $H=I_{m\ell}$. On the other hand, $\overline{\Omega}=\Omega$ whenever $(x^m-1) \mid \mbox{det}(\widetilde{G}(x))$, but $C$ is nontrivial unless each eigenvalue in $\Omega$ has multiplicity $\ell$.
\end{rem}

\begin{defn}\label{eigencode}
Let $\mathcal{V}\subseteq \F^\ell$ be an eigenspace. We define the {\it eigencode} corresponding to $\mathcal{V}$ by
$$\mathbb{C}(\mathcal{V})=\mathbb{C}:=\left\{\mathbf{u}\in \Fq^\ell\ : \ \sum_{j=0}^{\ell-1}{v_ju_j}=0, \forall \mathbf{v} \in \mathcal{V}\right\}.$$
When we have $\mathbb{C}=\{\mathbf{0}_{\ell}\}$, it is assumed that $d(\mathbb{C})=\infty$.
\end{defn}

The BCH-like minimum distance bound of Semenov and Trifonov for a given QC code in ~\cite[Theorem 2]{ST} was expressed in terms of the size of a consecutive subset of eigenvalues in $\overline{\Omega}$ and the minimum distance of the common eigencode related to this consecutive subset. Zeh and Ling generalized their approach and derived an HT-like bound in~\cite[Theorem 1]{LZ} without using the parity-check matrix in their proof. The eigencode, however, is still needed. In the next section, we will prove the analogues of these bounds for QC codes in terms of the Roos and shift bounds. 

\subsubsection{Spectral Bounds for QC Codes} \label{bounds section}
First, we establish a general spectral bound on the minimum distance of a given QC code. Let $C\subseteq\Fq^{m\ell}$ be a QC code of index $\ell$ with nonempty eigenvalue set $\overline{\Omega}\varsubsetneqq\Omega$. Let $P\subseteq\overline{\Omega}$ be a nonempty subset of eigenvalues such that $P=\{\xi^{u_1}, \xi^{u_2},\ldots,\xi^{u_r}\}$, where $0<r\leq|\overline{\Omega}|$.  We define 
\begin{equation*}\label{pmatrix}
\widetilde{H}_P:=\begin{pmatrix}
1&\xi^{u_1}&\xi^{2u_1}&\ldots&\xi^{(m-1)u_1}\\
\vdots & \vdots & \vdots & \vdots & \vdots \\
1&\xi^{u_r}&\xi^{2u_r}&\ldots&\xi^{(m-1)u_r}
\end{pmatrix}.
\end{equation*}
Recall that $d_P$ denotes a positive integer such that any cyclic code $C_P\subseteq\F_q^m$, whose zeros contain $P$, has a minimum distance at least $d_P$. We have $\widetilde{H}_P\mathbf{c}_P^{\top}=\mathbf{0}^{\top}_r$, for any $\mathbf{c}_P\in C_P$. In particular, if $P$ is equal to the set of zeros of $C_P$, then $\widetilde{H}_P$ is a parity-check matrix for $C_P$.

Let $\mathcal{V}_P$ denote the common eigenspace of the eigenvalues in $P$, and let $V_P$ be the matrix, say of size $t\times\ell$, whose rows form a basis for $\mathcal{V}_P$ (compare (\ref{Eigenspace})). If we set $\widehat{H}_P =\widetilde{H}_P \otimes V_P $, then $\widehat{H}_P \mathbf{c}^{\top}=\mathbf{0}^{\top}_{rt}$, for all $\mathbf{c}\in C$. In other words, $\widehat{H}_P$ is a submatrix of some $H$ of the form in (\ref{parity check matrix}) if $\mathcal{V}_P\neq\{\mathbf{0}_{\ell}\}$. If $\mathcal{V}_P=\{\mathbf{0}_{\ell}\}$, then $\widehat{H}_P$ does not exist. We first handle this case separately so that the bound is valid even if we have $\mathcal{V}_P=\{\mathbf{0}_{\ell}\}$, before the cases where we can use $\widehat{H}_P$ in the proof.

Henceforth, we consider the quantity $\min\{d_P,d(\mathbb{C}_P)\}$, where $\mathbb{C}_P$ is the eigencode corresponding to $\mathcal{V}_P$. We have assumed $P\neq\emptyset$ so that $\widetilde{H}_P$ is defined, and we also have $P\neq\Omega$ as $P\subseteq\overline{\Omega}\varsubsetneqq\Omega$ to make $d_P$ well-defined. Since $|P| \geq 1$, the BCH bound implies $d_P\geq 2$. Hence, $\min\{d_P,d(\mathbb{C}_P)\} =1$ if and only if $d(\mathbb{C}_P) =1$. For any nonzero QC code C, we have $d(C) \ge 1$. Therefore, whenever $d(\mathbb{C}_P) =1$, we have $d(C) \ge \min\{d_P,d(\mathbb{C}_P)\}$. In particular, when $\mathcal{V}_P={0}$, which implies $\mathbb{C}_P = \F_q^\ell$ and $d(\mathbb{C}_P) =1$, we have $d(C) \ge 1 = \min\{d_P,d(\mathbb{C}_P)\}$.

Now let $\emptyset\neq P\subseteq\overline{\Omega}\varsubsetneqq\Omega$ and let $d(\mathbb{C}_P)\geq 2$. Assume that there exists a codeword $\mathbf{c}\in C$ of weight $\omega$ such that $0 < \omega < \min\{d_P,d(\mathbb{C}_P)\}$. For each $0 \leq k \leq m-1$, let $\mathbf{c}_k = (c_{k,0}, . . . , c_{k,\ell-1})$ be the $k^{th}$ row of the codeword $\mathbf{c}$ given as in (\ref{array}) and we set $\mathbf{s}_k := V_P\mathbf{c}_k^{\top}$. Since $d(\mathbb{C}_P) > \omega$, we have $\mathbf{c}_k \notin \mathbb{C}_P$ and therefore $ \mathbf{s}_k=V_P\mathbf{c}_k^{\top}\neq \mathbf{0}^{\top}_t$, for all $\mathbf{c}_k \neq \mathbf{0}_{\ell}$, $k\in\{0,\ldots,m-1\}$. Hence, $0 < \lvert\{\mathbf{s}_k : \mathbf{s}_k \neq \mathbf{0}^{\top}_t \} \rvert \leq \omega < \min\{d_P,d(\mathbb{C}_P)\}$. Let $S := [\mathbf{s}_0\ \mathbf{s}_1 \cdots \mathbf{s}_{m-1}]$. Then $\widetilde{H}_PS^{\top} = 0$, which implies that the rows of the matrix $S$ lies in the right kernel of $\widetilde{H}_P$. But this is a contradiction since any row of $S$ has weight at most $\omega<d_P$, showing the following.

\begin{thm}[{\cite[Theorem 11]{ELOT}}]\label{main thm}
Let $C\subseteq R^\ell$ be a QC code of index $\ell$ with nonempty eigenvalue set $\overline{\Omega}\varsubsetneqq\Omega$. Let $P\subseteq\overline{\Omega}$ be a nonempty subset of eigenvalues and let $C_P\subseteq\F_q^m$ be any cyclic code with zeros $L\supseteq P$ and minimum distance at least $d_P$. We define $\mathcal{V}_P:=\bigcap_{\beta\in P}\mathcal{V}_{\beta}$ as the common eigenspace of the eigenvalues in $P$ and let $\mathbb{C}_P$ denote the eigencode corresponding to $\mathcal{V}_P$. Then, 
\begin{equation}\label{gen spectral}
d(C) \geq \min \left\{ d_P, d(\mathbb{C}_P) \right\}.
\end{equation}
\end{thm}

Theorem \ref{main thm} allows us to use any minimum distance bound derived for cyclic codes based on their zeros. The following special cases are immediate after the preparation that we have done in Section \ref{prepa} (cf. Theorems \ref{Roos} and \ref{shift bound}).

\begin{cor}[{\cite[Corollary 12]{ELOT}}]\label{Cor-Roos-Shift}
Let $C\subseteq R^\ell$ be a QC code of index $\ell$ with $\overline{\Omega}\varsubsetneqq\Omega$ as its nonempty set of eigenvalues.
\begin{enumerate}[leftmargin=*]
\item[i.] Let $N$ and $M$ be two nonempty subsets of $\Omega$ such that $MN\subseteq\overline{\Omega}$, where $MN:=\{\varepsilon\vartheta\mid \varepsilon\in M,\ \vartheta\in N\}$. If there exists a consecutive set $M'$ containing $M$ with $|M'|\leq |M|+d_N-2$, then $d(C)\geq \min(|M|+d_N-1,d(\mathbb{C}_{MN}))$.\vspace{3pt}

\item[ii.] For every $A\subseteq\Omega$ that is independent with respect to $\overline{\Omega}$, we have $d(C)\geq \min(|A|,d(\mathbb{C}_{T_A}))$, where $T_A:=A\cap \overline{\Omega}$.
\end{enumerate}
\end{cor}

\begin{proof}\hfill
\begin{enumerate}[leftmargin=*]
\item[i.] Let $N=\{\xi^{u_1},\ldots,\xi^{u_r}\}$ and $M=\{\xi^{v_1},\ldots,\xi^{v_s}\}$ be such that there exists a consecutive set $M'=\{\xi^z: v_1\leq z\leq v_s\}\subseteq\Omega$ containing $M$ with $|M'|\leq |M|+d_N-2$. We define the matrices 
\[\hspace{-10pt}\widetilde{H}_{N}:=\setlength\arraycolsep{4pt}\begin{pmatrix}
1&\xi^{u_1}&\xi^{2u_1}&\ldots&\xi^{(m-1)u_1}\\
\vdots & \vdots & \vdots & \vdots & \vdots \\
1&\xi^{u_r}&\xi^{2u_r}&\ldots&\xi^{(m-1)u_r}
\end{pmatrix},\ \widetilde{H}_{M}:=\setlength\arraycolsep{4pt}\begin{pmatrix}
1&\xi^{v_1}&\xi^{2v_1}&\ldots&\xi^{(m-1)v_1}\\
\vdots & \vdots & \vdots & \vdots & \vdots \\
1&\xi^{v_s}&\xi^{2v_s}&\ldots&\xi^{(m-1)v_s}
\end{pmatrix}. \]

Consider the joint subset $MN=\{\xi^{u_i+v_j} : 1\leq i\leq r, 1\leq j\leq s\} \subseteq\overline{\Omega}$. Let $B_k$ be the $k^{th}$ column of $\widetilde{H}_N$ for $k \in \{0,\ldots,m-1\}$. We create the joint matrix
\begin{equation*}
\widetilde{H}_{MN}=\setlength\arraycolsep{3pt}\begin{pmatrix}
B_0 &\xi^{v_1}B_1&\xi^{2v_1}B_2&\ldots&\xi^{(m-1)v_1}B_{m-1}\\
\vdots & \vdots & \vdots & \vdots & \vdots \\
B_0 &\xi^{v_s}B_1&\xi^{2v_s}B_2&\ldots&\xi^{(m-1)v_s}B_{m-1}
\end{pmatrix}.
\end{equation*}

Now let $\mathcal{V}_{MN}:=\bigcap_{\beta\in MN}\mathcal{V}_{\beta}$ denote the common eigenspace of the eigenvalues in $MN$ and let $V_{MN}$ be the matrix, whose rows form a basis for $\mathcal{V}_{MN}$, built as in (\ref{Eigenspace}). Let $\mathbb{C}_{MN}$ be the eigencode corresponding to $\mathcal{V}_{MN}$.  Setting $\widehat{H}_{MN} :=\widetilde{H}_{MN} \otimes V_{MN}$ implies $\widehat{H}_{MN} \mathbf{c}^{\top}=\mathbf{0}$ for all $\mathbf{c}\in C$. The rest of the proof is identical with the proof of Theorem \ref{main thm}, where $P$ is replaced by $MN$, and the result follows by the Roos bound (Theorem \ref{Roos}).

\item[ii.] For each independent $A\subseteq\Omega$ with respect to $\overline{\Omega}$, let $T_A=\{\xi^{w_1},\xi^{w_2},\ldots,\xi^{w_y}\}=A\cap \overline{\Omega}$. Since $\overline{\Omega}$ is a proper subset of $\Omega$, a nonempty $T_A$ can be obtained by the recursive construction of $A$. We define 
$$\widetilde{H}_{T_A}=\begin{pmatrix}
1&\xi^{w_1}&\xi^{2w_1}&\ldots&\xi^{(m-1)w_1}\\
\vdots & \vdots & \vdots & \vdots  & \vdots \\
1&\xi^{w_y}&\xi^{2w_y}&\ldots&\xi^{(m-1)w_y}
\end{pmatrix}.$$
Let $V_{T_A}$ be the matrix corresponding to a basis of $\mathcal{V}_{T_A}$, which is the intersection of the eigenspaces belonging to the eigenvalues in $T_A$. Let $\mathbb{C}_{T_A}$ be the eigencode corresponding to the eigenspace $\mathcal{V}_{T_A}$. We again set $\widehat{H}_{T_A} := \widetilde{H}_{T_A} \otimes V_{T_A}$ and the result follows in a similar way by using the shift bound (Theorem \ref{shift bound}).
\end{enumerate}
\end{proof}

\begin{rem}
We can obtain the BCH-like bound in \cite[Theorem 2]{ST} and the HT-like bound in \cite[Theorem 1]{LZ} by using Remarks \ref{Roos remark} and \ref{shift remark}.
\end{rem}

\section{Asymptotics}\label{asymptotics}
Let $(C_i)_{i \geq 1}$ be a sequence of linear codes over $\mathbb{F}_{q}$, and let $N_i, d_i $ and $k_i$ denote respectively the length, minimum distance, and dimension of $C_i$, for all $i$. Assume that $\displaystyle\lim_{i\rightarrow \infty} N_i=\infty$. Let
\[
\delta := \displaystyle{\liminf_{i\rightarrow \infty}\frac{d_i}{N_i}} \mbox{\ \ and\ \ } \mathcal{R} := \displaystyle{\liminf_{i\rightarrow \infty}\frac{k_i}{N_i}}
\]
denote the relative distance and the relative rate of the sequence $(C_i)_{i \geq 1}$. Both $\mathcal{R}$ and $\delta$ are finite as they are limits of bounded quantities. If $\mathcal{R}\delta \neq 0$, then $(C_i)_{i \geq 1}$ is called an \emph{asymptotically good sequence} of codes.

We will require the celebrated entropy function
\[
H_q(y)=y\log_q(q-1)-y\log_q(y)-(1-y)\log_q(1-y),
\]
defined for $0<y<\frac{q-1}{q}$ and of constant use in estimating binomial coefficients of large arguments \cite[pages 309--310]{MS77}. The asymptotic Gilbert--Varshamov Bound (see \cite[Chapter 17, Theorem 30]{MS77}) states that, for every $q$ and $0< \delta < 1-\frac{1}{q}$, there exists an infinite family of $q$-ary codes with limit rate
\[
\mathcal{R} \geq 1- H_q(\delta).
\]

The existence of explicit families of QC codes satisfying a modified Gilbert--Varshamov Bound were shown in  \cite{CPW}, and then this result was improved in \cite{K}. Below, we focus on self-dual and LCD families of QC codes.

\subsection{Good Self-Dual QC Codes Exist}
In this section, we construct families of binary self-dual QC codes. We assume that all binary codes are equipped with the Euclidean inner product and all the $\mathbb{F}_4$-codes are equipped with the Hermitian inner product. Self-duality in the following discussion is with respect to these respective inner products. A binary self-dual code is said to be of \emph{Type II} if and only if all its weights are multiples of $4$ and of \emph{Type I} otherwise. We first recall some background material on mass formulas for self-dual codes over $\mathbb{F}_2$ and $\mathbb{F}_4$. 
\begin{prop}\label{prelim}
Let $\ell$ be an even positive integer.
\begin{enumerate}[leftmargin=*]
\item[i.] The number of self-dual binary codes of length $\ell$ is given by $$N(2,\ell) = \prod_{i=1}^{\frac{\ell}{2} -1} (2^i +1). $$
\item[ii.] Let $\mathbf{v}\in\F_2^\ell$ with even Hamming weight, other than $0$ and $\ell$. The number of self-dual binary codes of length $\ell$ containing $\mathbf{v}$ is given by $$M(2, \ell) = \prod_{i=1}^{\frac{\ell}{2} -2} (2^i +1).$$
\item[iii.] The number of self-dual $\F_4$-codes of length $\ell$ is given by $$N(4,\ell) = \prod_{i=0}^{\frac{\ell}{2} -1} (2^{2i+1} +1).$$
\item[iv.] The number of self-dual $\F_4$-codes of length $\ell$ containing a given nonzero codeword of length $\ell$ and even Hamming weight is given by $$M(4, \ell) = \prod_{i=0}^{\frac{\ell}{2} -2} (2^{2i+1} +1).$$
\end{enumerate} 
\end{prop} 
\begin{proof}
Parts $i.$ and $iii.$ are well-known facts, found for example in \cite{rs}. Part $ii.$ is an immediate consequence of \cite[Theorem 2.1]{mst} with $s=2$, noting that every self-dual binary code must contain the all-one vector $\mathbf{1}$. Part $iv.$  follows from \cite[Theorem 1]{cps} with $n_1 = \ell$ and $k_1 =1$.
\end{proof}

\begin{prop}\label{ptype2}
Let $\ell$ be a positive integer divisible by $8$.
\begin{enumerate}[leftmargin=*]
\item[i.] The number of Type II binary self-dual codes of length $\ell$ is given by $$T(2,\ell) = 2 \prod_{i=1}^{\frac{\ell}{2} -2} (2^i +1).$$
\item[ii.] Let $\mathbf{v}\in\F_2^\ell$ with Hamming weight divisible by $4$, other than $0$ and $\ell$. The number of Type II binary self-dual codes of length $\ell$ containing $\mathbf{v}$ is given by $$S(2, \ell) = 2 \prod_{i=1}^{\frac{\ell}{2} -3} (2^i +1).$$
\end{enumerate} 
\end{prop} 

\begin{proof}
Part $i.$ is found in \cite{rs}, and part $ii.$ is exactly \cite[Corollary 2.4]{mst}.
\end{proof}

Let $C_1$ denote a binary code of length $\ell$ and let $C_2$ be a code over $\F_4$ of length $\ell$. We construct a binary QC code $C$ of length $3\ell$ and index $\ell$ whose codewords are of the form given in (\ref{cubic}). It is easy to check that $C$ is self-dual if and only if both $C_1$ and $C_2$ are self-dual, and $C$ is of Type II if and only if $C_1$ is of Type II and $C_2$ is self-dual.

We assume henceforth that $C$ is a binary self-dual QC code constructed in the above way. Any codeword $\mathbf{c} \in C$ must necessarily have even Hamming weight. Suppose that $\mathbf{c}$ corresponds to the pair $(\mathbf{c}_1, \mathbf{c}_2)$, where $\mathbf{c}_1 \in C_1$ and $\mathbf{c}_2 \in C_2$. Since $C_1$ and $C_2$ are self-dual, it follows that $\mathbf{c}_1$ and $\mathbf{c}_2$ must both have even Hamming weights. When $\mathbf{c} \neq \mathbf{0}_{3\ell}$, there are three possibilities for $(\mathbf{c}_1, \mathbf{c}_2)$:
\begin{enumerate}[leftmargin=*]
\item $ \mathbf{c}_1 \neq \mathbf{0}_{\ell}$, $\mathbf{c}_2 \neq \mathbf{0}_{\ell}$;
\item $\mathbf{c}_1 = \mathbf{0}_{\ell}$, $\mathbf{c}_2 \neq \mathbf{0}_{\ell}$; and
\item $\mathbf{c}_1 \neq \mathbf{0}_{\ell}$, $\mathbf{c}_2 = \mathbf{0}_{\ell}$.
\end{enumerate} 
We count the number of codewords $\mathbf{c}$ in each of these categories for a given even weight $d$.

For Case 2, if the Hamming weight of $\mathbf{c}$ is $d$, then $C\cong C_1 \oplus C_2$ implies that $\mathbf{c}_2$ has Hamming weight $d/2$. Since $\mathbf{c}_2$ has even Hamming weight, it follows that $d$ is divisible by 4 in order for this case to occur. It is easy to see that the number $A_2 (\ell, d)$ of such words $\mathbf{c}$ is bounded above by $\binom{\ell}{d/2} 3^{d/2}$ where $4\,|\,d$. For $d$ not divisible by 4, set $A_2 (\ell , d) =0$.

The argument to obtain the number of words of Case 3 is similar. It is easy to show that the number $A_3 (\ell , d)$ of such words is bounded above  by $\binom{\ell}{d/3}$ where $6\,|\,d$. When $d$ is not divisible by 6, set $A_3(\ell,d) =0$.

The total number of vectors in $\mathbb{F}_2^{3 \ell}$ of weight $d$ is $\binom{3 \ell}{d}$; so for Case 1, we have
\[
A_1 (\ell, d) \leq \binom{3 \ell}{d} - A_2(\ell, d) - A_3 (\ell, d).
\]
In particular, $A_1(\ell , d) $ is bounded above by $\binom{3 \ell}{d}$.

Combining the above observations and Proposition \ref{prelim}, the number of self-dual binary QC codes of length $3 \ell$ and index $\ell$ with minimum weight less than $d$ is bounded above by
\[
\sum_{\substack{e<d \\ e : \mbox{\scriptsize\,even}}} \bigl(A_1(\ell , e) M(2, \ell) M(4 , \ell) + A_2 (\ell , e) N(2, \ell) M(4, \ell) + A_3 (\ell, e) M(2, \ell) N(4, \ell) \bigr).
\]

\begin{thm}[{\cite[Theorem 3.1]{LS2}}]\label{finite}
Let $\ell$ be an even integer and let $d$ be the largest even integer such that
\begin{align*}
\sum_{\substack{e<d\\ e\, \equiv\, 0 \bmod{2}}} \binom{3 \ell}{e} +  (2^{\frac{\ell}{2} -1} +1)\Biggl(& \sum_{\substack{e<d\\ e\, \equiv\, 0 \bmod{4}}} \binom{\ell}{e/2} 3^{e/2}\Biggr) + (2^{\ell -1} +1)\Biggl( \sum_{\substack{e<d\\ e\, \equiv\, 0 \bmod{6}}} \binom{\ell}{e/3}\Biggr)  \\[7pt]
&< (2^{\frac{\ell}{2} -1} +1)(2^{\ell -1} +1).
\end{align*}
Then there exists a self-dual binary QC code of length $3 \ell$ and index $\ell$ with minimum distance at least $d$.
\end{thm} 
\begin{proof}
Multiplying both sides of the inequality by $M(2, \ell) M(4, \ell)$, and applying the above upper bounds for $A_i (\ell ,d)$ ($i=1,2,3$), we see that the inequality in Theorem \ref{finite} implies that the number of self-dual binary QC codes of length $3\ell$ and index $\ell$ with minimum distance $<d$ is strictly less than the total number of self-dual binary QC codes of length $3 \ell$ and index $\ell$.
\end{proof}

If we are interested only in Type II QC codes, using Proposition \ref{ptype2}, we easily see that the number of Type II binary QC codes of length $3 \ell$ and index $\ell$ with minimum weight is $<d$  is bounded above by
\[
\sum_{\substack{e<d \\ e\, \equiv\, 0 \bmod{4}}} \bigl(A_1(\ell , e) S(2, \ell) M(4 , \ell) + A_2 (\ell , e) T(2, \ell) M(4, \ell) + A_3 (\ell, e) S(2, \ell) N(4, \ell) \bigr).
\]

Using an argument similar to that for Theorem \ref{finite}, we obtain the following result.
\begin{thm}[{\cite[Theorem 3.2]{LS2}}]\label{ftype2}
Let $\ell$ be divisible by $8$ and let $d$ be the largest multiple of $4$ such that
\begin{align*}
\sum_{\substack{e<d\\ e\, \equiv\, 0 \bmod{4}}} \binom{3 \ell}{e} + (2^{\frac{\ell}{2} -2} +1)\Biggl(& \sum_{\substack{e<d\\ e\, \equiv\, 0 \bmod{4}}} \binom{\ell}{e/2} 3^{e/2}\Biggr)  + (2^{\ell -1} +1)\Biggl(\sum_{\substack{e<d\\ e\, \equiv\, 0 \bmod{12}}} \binom{\ell}{e/3}\Biggr) \\[7pt]
&< (2^{\frac{\ell}{2} -2} +1) (2^{\ell -1} +1).
\end{align*}
Then there exists a Type II binary QC code of length $3 \ell$ and index $\ell$ with minimum distance at least $d$.
\end{thm}

We are now in a position to state and prove the asymptotic version of Theorems \ref{finite} and \ref{ftype2}.
\begin{thm}[{\cite[Theorem 4.1]{LS2}}] There exists an infinite family of binary self-dual QC codes $\C_i$ of length $3\ell_i$ and of distance $d_i$ with $\displaystyle\lim_{i\rightarrow\infty}\ell_i=\infty$ such that $\displaystyle\delta=\liminf_{i\rightarrow\infty}\frac{d_i}{3\ell_i}$ exists and is bounded below by
\[
\delta \ge H_2^{-1}(1/2)=0.110\cdots.
\]
\end{thm}
\begin{proof}
The right-hand-side of the inequality of Theorem \ref{finite} is plainly of the order of $2^{3\ell/2}$ for large $\ell$. We compare this in turn to each of the three summands on the left-hand side (at the price of a more stringent inequality, congruence conditions on the summation range are neglected). By \cite[Chapter 10, Corollary 9]{MS77}, for large $\ell$ (with $\mu =\delta$ and $n=\ell$), the first and third summands are of order $2^{3\ell H_2(\delta)}$ and $2^{\ell +\ell H_2(\delta)}$, respectively. They both are of the order of the RHS for $H_2(\delta)=1/2$. By \cite[Chapter 10, Lemma 7]{MS77}, for large $\ell$ (with $\lambda=\delta$ and $n=\ell$), the second summand is of order $2^{\ell f(3\delta/2)}$ for $f(t):=0.5+t\log_2(3)+H_2(t)$, which is of the order of the right-hand-side for $$\delta=0.1762\cdots,$$ a value $>H_2^{-1}(1/2)$.
\end{proof}

Similarly, for Type II codes, we have the following.
\begin{thm} There exists an infinite family of Type II binary QC codes $C_i$ of length $3\ell_i$ and of distance $d_i$  with $\displaystyle\lim_{i\rightarrow\infty}\ell_i=\infty$ such that $\displaystyle\delta=\liminf_{i\rightarrow\infty}\frac{d_i}{3\ell_i}$ exists and is bounded below by
\[
\delta \ge H_2^{-1}(1/2)=0.110\cdots .
\]
\end{thm}
\begin{proof}
Since we neglected the congruence conditions in the preceding analysis, the calculations are exactly the same but using Theorem~\ref{ftype2}.
\end{proof}

\subsection{Complementary-Dual QC Codes Are Good}
 
It is known that both Euclidean and Hermitian (using the inner product \eqref{hermprod}) LCD codes are asymptotically good; see \cite[Propositions 2 and 3]{M} and \cite[Theorem 3.6]{GOS}, respectively. Recall that the second part of Corollary \ref{CDinstance} suggests an easy construction of QCCD codes from LCD codes. In \cite{GOS}, this idea was used to show the existence of good long QCCD codes, as shown in the next result, together with the help of Theorem \ref{jensen}.

\begin{thm}[{\cite[Theorems 3.3 and 3.7]{GOS}}] \label{asymptotics-LCD}
Let $q$ be a power of a prime and let $m\geq 2$ be relatively prime to $q$. Then there exists an asymptotically good sequence of $q$-ary QCCD codes.
\end{thm}

\begin{proof}Let $\xi$ be a primitive $m^{\mathrm{th}}$ root of unity over $\mathbb{F}_q$. There are two possibilities.
\begin{enumerate}
\item[(1)] $1$ and $-1$ are not in the same $q$-cyclotomic coset modulo $m$.
\item[(2)] $1$ and $-1$ are in the same $q$-cyclotomic coset modulo $m$.
\end{enumerate}
In case (1), $\xi$ and $\xi^{-1}$ have distinct minimal polynomials $h'(x)$ and $h''(x)$, respectively, over $\mathbb{F}_q$. Let $\mathbb{H}'=\mathbb{F}_q[x]/\langle h'(x) \rangle$ and $\mathbb{H}''=\mathbb{F}_q[x]/\langle h''(x) \rangle$. Recall that these fields are equal: $\mathbb{H}'=\mathbb{F}_q(\xi)=\mathbb{F}_q(\xi^{-1})=\mathbb{H}''$. We denote both by $\mathbb{H}$. Let $\theta'$ and $\theta''$ denote the primitive idempotents corresponding to the $q$-ary length $m$ minimal cyclic codes with check polynomials $h'(x)$ and $h''(x)$, respectively. Let $(C_i)_{i\geq 1}$ be an asymptotically good sequence of (Euclidean) LCD codes over $\mathbb{H}$ (such a sequence exists by \cite{M}), where each $C_i$ has parameters $[\ell_i,k_i,d_i]$. For all $i\geq 1$, define the $q$-ary QC code $D_i$ as
\[
%\label{sequence}
D_i:=C_i\oplus C_i=\bigl( \langle \theta' \rangle \Box C_i
\bigr)\oplus \bigl( \langle \theta'' \rangle \Box C_i \bigr) \subset \mathbb{H}^{\ell_i}\oplus \mathbb{H}^{\ell_i}.
\]
By Corollary \ref{CDinstance}, each $D_i$ with $i\geq 1$ is a QCCD code of index $\ell_i$. If $e:=[\mathbb{H}:\mathbb{F}_q]$, then the length and the $\mathbb{F}_q$-dimension of $D_i$ are $m\ell_i$ and $2ek_i$, respectively. By the Jensen Bound of Theorem \ref{jensen}, the minimum distance of $D_i$ satisfies
\[
d(D_i)\geq \min\,\{d(\langle \theta' \rangle)d_i,d(\langle \theta' \rangle \oplus \langle \theta'' \rangle)d_i \}\geq d(\langle \theta' \rangle \oplus \langle \theta'' \rangle)d_i.
\]
For the sequence of QCCD codes $(D_i)_{i\geq 1}$, the relative rate is
\[
\mathcal{R}=\liminf_{i\to \infty} \frac{2ek_i}{m\ell_i}=\frac{2e}{m}\liminf_{i\to \infty} \frac{k_i}{\ell_i},
\]
and this quantity is positive since $(C_i)_{i\geq 1}$ is asymptotically good. For the relative distance, we have
\[
\delta=\liminf_{i\to \infty} \frac{d(D_i)}{m\ell_i}\geq d(\langle \theta' \rangle \oplus \langle \theta'' \rangle) \liminf_{i\to \infty} \frac{d_i}{\ell_i}.
\]
Note again that $\delta$ is positive since $(C_i)_{i\geq 1}$ is asymptotically good.

In case (2), we clearly have that $n$ and $-n$ are always in the same $q$-cyclotomic coset modulo $m$. Therefore every irreducible factor of $x^m-1$ over $\mathbb{F}_q$ is self-reciprocal. As $m\geq 2$, we can choose such a factor $g(x)\ne x-1$. Let $\mathbb{G}=\mathbb{F}_q[x]/\langle g(x) \rangle$. We denote by $\theta$ the primitive idempotent corresponding to the $q$-ary minimal cyclic code of length $m$ with check polynomial $g(x)$. Let $(C_i)_{i\geq 1}$ be an asymptotically good sequence of Hermitian LCD codes over $\mathbb{G}$. Such a sequence exists by \cite[Theorem 3.6]{GOS}. Assume that each $C_i$ has parameters $[\ell_i,k_i,d_i]$. For each $i\geq 1$, define the $q$-ary QC code $D_i$ as the QC code with one constituent:
\[
%\label{GLO:sequence1}
D_i:=\langle \theta \rangle \Box C_i.
\]
If $e:=[\mathbb{G}:\mathbb{F}_q]$, then the length and the dimension of $D_i$ are $m\ell_i$ and $ek_i$, respectively. By the Jensen Bound of Theorem \ref{jensen}, we have
\[
d(D_i)\geq d(\langle \theta \rangle)d(C_i).
\]
As in part (a) we conclude that $(D_i)_{i\geq 1}$ is asymptotically good since $(C_i)_{i\geq 1}$ is asymptotically good.
\end{proof}

\section{Connection to Convolutional Codes}
An $(\ell,k)$ convolutional code $C$ over $\F_q$ is defined as a rank $k$ $\Fq[x]$-submodule of $\Fq[x]^{\ell}$, which is necessarily a free module since $\F_q[x]$ is a principal ideal domain. The weight of a polynomial $c(x)\in \F_q[x]$ is defined as the number of nonzero terms in $c(x)$ and the weight of a codeword $\mathbf{c}(x)=(c_0(x),\ldots,c_{\ell-1}(x)) \in C$  is the sum of the weights of its coordinates. The free distance of the convolutional code $d_{\text{free}}(C)$ is the minimum weight among its nonzero codewords.
\begin{rem}\label{note on gen matrix}
An encoder of an $(\ell,k)$ convolutional code $C$ is a $k\times\ell$ matrix $G$ of rank $k$ with entries from $\F_q[x]$. In other words,
\begin{equation*} \label{conv code}
C=\left\{ \left(u_0(x),\ldots ,u_{k-1}(x)\right)G: \
\left(u_0(x),\ldots ,u_{k-1}(x)\right)\in \F_q[x]^k \right\}.
\end{equation*}
\end{rem}

\begin{rem}\label{note on conv}
Note that a convolutional code is in general defined as an $\F_q(x)$-submodule (subspace) of $\F_q(x)^\ell$. However, this leads to codewords with rational coordinates and infinite weight. From the practical point of view, there is no reason to use this as the definition (see \cite{GRS,L}). Note that even if $C$ is defined as a subspace of $\F_q(x)^\ell$, it has an encoder which can be obtained by clearing off the denominators of all the entries in any encoder. Moreover, it is usually assumed that $G$ is noncatastrophic in the sense that a finite weight codeword $\mathbf{c}(x)\in C$ can only be produced from a finite weight information word $\mathbf{u}(x)$. In other words, an encoder $G$ is said to be noncatastrophic if, for any $\mathbf{u}(x) \in \F_q(x)^k$, $\mathbf{u}(x)G$ has finite weight implies that $\mathbf{u}(x)$ also has finite weight. Hence, with a  noncatastrophic encoder $G$, all finite weight codewords are covered by the $\F_q[x]$-module structure. Noncatastrophic encoders exist for any
convolutional code (see \cite{McE}).
\end{rem}

Let $ R=\mathbb{F}_q[x]/ \langle x^m-1\rangle$ as before and consider the projection map
\[
%\label{GLO:projection1}
{\setlength\arraycolsep{3pt}\begin{array}{rll}
\Psi:\mathbb{F}_q[x] &\longrightarrow& R \\[1ex]
 f(x)& \longmapsto & f'(x):=f(x)\!\!\!\mod \langle x^m-1\rangle.
\end{array}}
\]
It is clear that for a given $(\ell,k)$ convolutional code $C$ and any $m>1$, there is a natural QC code $C'$, of length $m\ell$ and index $\ell$, related to it as shown below. Note that we denote the map from $C$ to $C'$ also by $\Psi$.
\[
%\label{GLO:1Drelation}\label{GLO:projection2}
{\setlength\arraycolsep{4pt}\begin{array}{rll}
\Psi: C &\longrightarrow& C'\\[1ex]
\mathbf{c}(x)=(c_0(x),\ldots,c_{\ell-1}(x))&\longmapsto& \mathbf{c'}(x)=(c'_0(x),\ldots,c'_{\ell-1}(x)).
\end{array}}
\]

Lally \cite{L} showed that the minimum distance of the QC code $C'$ above is a lower bound on the free distance of the convolutional code $C$.

\begin{thm}[{\cite[Theorem 2]{L}}] \label{distance}
Let $C$ be an $(\ell,k)$ convolutional code over $\mathbb{F}_q$ and let $C'$ be the related QC code in $R^\ell$. Then $d_{\text{free}}(C)\geq d(C')$.
\end{thm}
\begin{proof} 
Let $\mathbf{c}(x)$ be a codeword in $C$ and set $\mathbf{c'}(x)=\Psi(\mathbf{c}(x)) \in C'$. We consider the two possibilities.
First, if $\mathbf{c'}(x)\not= \mathbf{0}$, then $\mathrm{wt}(\mathbf{c}(x))\geq \mathrm{wt}(\mathbf{c'}(x))$. For the second possibility, suppose $\mathbf{c'}(x)= \mathbf{0}$. Let $\gamma \geq 1$ be the maximal positive integer such that $(x^m-1)^\gamma$ divides each coordinate of $\mathbf{c}(x)$. Write $\mathbf{c}(x)=(x^m-1)^\gamma \left(v_0(x),\ldots ,v_{\ell-1}(x) \right)$ and set $\mathbf{v}(x)=\left(v_0(x),\ldots ,v_{\ell-1}(x) \right)$. Then $\mathbf{v}(x)$ is a codeword of $C$, and for $\mathbf{v'}(x) \in C'\setminus \{\mathbf{0}\}$, we have $\mathrm{wt}(\mathbf{c}(x))\geq \mathrm{wt}(\mathbf{v'}(x))$ by the first possibility. Combining the two possibilities, for any $\mathbf{c}(x)\in C$, there exists a $\mathbf{v'}(x)\in C'$ such that $\mathrm{wt}(\mathbf{c}(x))\geq \mathrm{wt}(\mathbf{c'}(x))$. This proves that $d_{\text{free}}(C)\geq d(C')$.
\end{proof}

\begin{rem}
Note that Lally uses an alternative module description of convolutional and QC codes in \cite{L}. Namely, a basis $\{1,\alpha,\ldots ,\alpha^{\ell-1}\}$ of $\F_{q^\ell}$ over $\F_q$ is fixed and the $\F_q[x]$-modules $\F_q[x]^\ell$ and $\F_{q^\ell}[x]$ are identified via the map $\Phi$ in (\ref{Lal}). With this identification, a length $\ell$ convolutional code is viewed as an $\F_q[x]$-module in $\mathbb{F}_{q^{\ell}}[x]$ and a length $m\ell$, index $\ell$ QC code is viewed as an $\F_q[x]$-module in $\mathbb{F}_{q^{\ell}}[x]/\langle x^m-1\rangle$. However, all of Lally's findings can be translated to the module descriptions that we have been using for convolutional and QC codes and this is how they are presented in Theorem \ref{distance}.
\end{rem}

\section*{Acknowledgements}
The third author thanks Frederic Ezerman, Markus Grassl and Patrick Sol\'e for their valuable comments and discussions. The authors are also grateful to Cary Huffman for his great help in improving this chapter. The second and third authors are supported by NTU Research Grant M4080456.\nobreak

\bibliographystyle{IEEEtranS}
\bibliography{QC-chapter} 

\end{document}